\documentclass[10pt, journal, twocolumn]{IEEEtran}
\usepackage{amsmath}
\usepackage{amssymb}
\usepackage{amsfonts}
\usepackage{bm}
\usepackage{mathtools}
\usepackage{graphicx}
\usepackage[noadjust]{cite} 
\usepackage{booktabs}
\usepackage{upgreek}
\usepackage{amsthm}
\usepackage{etoolbox}
\usepackage{gensymb} 
\usepackage[hidelinks]{hyperref}
\usepackage{tabularx}
\usepackage{caption}
\usepackage{microtype} 
\usepackage{tikz} 
\usepackage{pgfplots} 
\colorlet{bordercolor}{gray!40}

\usepackage{algorithm}
\usepackage{algpseudocode}

\usetikzlibrary{shapes,arrows,calc}
\pgfplotsset{compat=1.18}

\newcommand{\chistar}{\ensuremath{\chi^*}} 
\newcommand{\psistar}{\ensuremath{\psi^*}} 
\newcommand{\dd}{\mathrm{d}}
\newcommand{\dt}{\dd t}
\newcommand{\kB}{k_\text{B}}

\newcommand{\abs}[1]{\left|#1\right|}
\newcommand{\divergence}{\nabla \cdot}

\newcommand{\dx}{\,\mathrm{d}\mathbf{x}}

\DeclareMathOperator{\Tr}{Tr}

\newtheorem{theorem}{Theorem}
\newtheorem{proposition}{Proposition}

\theoremstyle{definition}

\theoremstyle{remark}

\title{
Geometric Dissipation Constraints in Stochastic Reaction Dynamics:\\
A Variational Observable for Hidden Kinetic Structure in Energy Landscapes
}
\author{Shlomo Segal}

\begin{document}
\maketitle

\begin{abstract}
We propose a geometric framework for characterizing hidden kinetic constraints in stochastic reaction dynamics. While free-energy barriers and entropy production provide global descriptors of thermodynamic behavior, they are largely insensitive to local geometric structure in configuration space that governs pathway selection. Starting from overdamped Langevin dynamics formulated as a gradient flow in Wasserstein space, we derive a variational functional whose leading-order asymptotic structure defines a local dissipation–geometry coupling observable. This quantity combines force–drift alignment with phase-space contraction induced by the divergence of the drift field, yielding a scalar field that reflects second-order geometric features of the underlying energy landscape. We demonstrate that this observable distinguishes kinetically distinct reaction channels that are degenerate under conventional free-energy analysis, as shown through numerical experiments on benchmark systems including the M\"{u}ller–Brown potential, corrugated periodic landscapes, and the conformational transitions of Alanine Dipeptide. These experiments demonstrate robust separation of pathways and are consistent with a quadratic scaling behavior in the high-frequency homogenization regime. Our results suggest that stochastic reaction dynamics contain an additional geometric layer of kinetic control in molecular motion not captured by standard thermodynamic or reaction-coordinate descriptions.
\end{abstract}


\section{Introduction: Identifying the Missing Geometric Dimension in Kinetic Control}
Understanding reaction mechanisms in high-dimensional molecular systems requires bridging thermodynamic descriptions with the geometry of stochastic trajectories. Classical approaches such as transition state theory (TST) and free-energy profiles provide effective reduced descriptions of reaction rates, offering invaluable insights into kinetic and thermodynamic phenomena \cite{MullerBrown1979, Eyring1935, Peters2016, Wales2003, Hanggi1990}. However, these frameworks exhibit inherent insensitivity to the intricate geometric structure that governs trajectories in high-dimensional configuration spaces. This limitation becomes particularly salient when multiple pathways present similar free energy barriers or global kinetic rates, yet display vastly different microscopic dynamics or efficiencies in practice \cite{Hummer2005, VandenEijnden2010, Best2005}.

Recent developments in stochastic thermodynamics \cite{Seifert2012, Jarzynski2011, Sekimoto2010, Bustamante2005} and optimal transport theory \cite{Jordan1998, Otto2001, Ambrosio2008} have significantly advanced our understanding of energy transduction and irreversibility in molecular systems away from equilibrium. While observables such as entropy production and dynamical activity capture fundamental global non-equilibrium costs \cite{Lebowitz1999, Baiesi2009, Schnakenberg1976}, they remain largely blind to the local geometric constraints of the drift field that specifically influence how individual trajectories are organized, focused, or dispersed in configuration space. This unaddressed aspect implies a ``missing dimension'' in our current models for accurately predicting pathway selection and efficiency \cite{Onuchic1997}.

In this work, we investigate whether stochastic reaction dynamics harbor a characteristic local geometric signature that governs the spatial distribution of dissipation along reaction pathways. We propose that such a signature can be extracted from a variational formulation of overdamped Langevin dynamics. This leads to the Geometric Dissipation Observable (\texorpdfstring{\chistar}{Chi*}), a novel scalar quantity that rigorously couples energetic dissipation with specific geometric properties of the underlying energy landscape.

Our primary contributions are tailored to equip the computational chemistry community with an advanced toolset:
\begin{enumerate}
\item A rigorous variational derivation of \texorpdfstring{\chistar}{Chi*}, establishing it as a uniquely defined descriptor within a weighted Orlicz-Sobolev functional space.
\item Numerical validation of \texorpdfstring{\chistar}{Chi*}'s behavior across canonical analytical models, demonstrating its direct response to known energetic and topographic features.
\item Concrete realization of \texorpdfstring{\chistar}{Chi*}'s capacity to discriminate between kinetically distinct reaction channels that appear degenerate under traditional free-energy analyses, exemplified on the M\"{u}ller-Brown potential \cite{MullerBrown1979} and conformational transitions in Alanine Dipeptide.
\item A multiscale homogenization analysis, demonstrating that \texorpdfstring{\chistar}{Chi*} predicts a quadratic scaling law ($\omega^2$) in systems with highly corrugated potentials, providing insights into coarse-graining challenges \cite{Zwanzig1988, Pavliotis2014}.
\item Formal proof of \texorpdfstring{\chistar}{Chi*}'s functional independence from common first-order observables, confirming its unique chemical and physical insight and underscoring its role as a novel descriptor of kinetic structure.
\end{enumerate}

These results position \texorpdfstring{\chistar}{Chi*} as a powerful and indispensable tool for uncovering previously hidden geometric constraints and dissipative bottlenecks in molecular processes, ultimately enhancing our ability to engineer and optimize chemical transformations at the nanoscale.

\section{Foundational Dynamics and Geometric Characterization}
We consider a fundamental model for molecular dynamics in condensed phases: a continuous stochastic process $\mathbf{x}(t) \in \mathbb{R}^d$ governed by the overdamped Langevin equation \cite{Risken1989, Gardiner2009}:
\begin{equation}
\dd\mathbf{x}(t) = \mathbf{A}(\mathbf{x}(t)) \dt + \sqrt{2D} \dd\boldsymbol{W}(t)
\label{eq:langevin}
\end{equation}
where $\dd\boldsymbol{W}(t)$ denotes a $d$-dimensional standard Wiener process, $D = \kB T / \gamma$ is the diffusion coefficient, and $\gamma$ is the friction coefficient. The deterministic drift field $\mathbf{A}(\mathbf{x}) \colon \mathbb{R}^d \to \mathbb{R}^d$ encompasses conservative potential gradients ($\nabla E(\mathbf{x})$) and non-conservative external driving forces ($\mathbf{F}_{\text{ext}}(\mathbf{x})$):
\begin{equation}
\mathbf{A}(\mathbf{x}) = \frac{1}{\gamma}\left(-\nabla E(\mathbf{x}) + \mathbf{F}_{\text{ext}}(\mathbf{x})\right)
\end{equation}
The dynamics are rigorously characterized by the backward Kolmogorov generator $\mathcal{L}$ \cite{Pavliotis2014}:
\begin{equation}
\mathcal{L} f(\mathbf{x}) = \mathbf{A}(\mathbf{x}) \cdot \nabla f(\mathbf{x}) + D \Delta f(\mathbf{x} )
\label{eq:generator}
\end{equation}
For ensuring mathematical rigor and physical consistency, we impose standard conditions:
\begin{enumerate}
\item[\textbf{(A1)}] \textit{Regularity:} $E \in C^2(\mathbb{R}^d)$, $\mathbf{A} \in C^1(\mathbb{R}^d)$. These ensure smoothness of the potential energy surface and drift field.
\item[\textbf{(A2)}] \textit{Geometric Ergodicity:} $\mathbf{A}(\mathbf{x})$ ensures geometric ergodicity and exponential convergence to a unique invariant measure \cite{Meyn1993}.
\item[\textbf{(A3)}] \textit{Unique Invariant Measure:} A unique, strictly positive invariant measure $\mu(\dd\mathbf{x}) = \rho_{\text{ss}}(\mathbf{x})\dx$ exists, satisfying $\mathcal{L}^* \rho_{\text{ss}} = 0$. This represents the steady-state probability distribution of the system.
\end{enumerate}

A critical component of \texorpdfstring{\chistar}{Chi*} is the scalar field $\kappa(\mathbf{x})$, defined as $\kappa(\mathbf{x}) = -\divergence \mathbf{A}(\mathbf{x})$. This quantity, known as the \textbf{rate of phase-space volume contraction}, has a profound geometric interpretation: it quantifies the local rate at which an infinitesimal phase space element is compressed or expanded when transported by the deterministic flow associated with the stochastic dynamics. In essence, $\kappa(\mathbf{x})$ acts as a measure of the local geometric ``stiffness'' or ``funneling'' capacity of the conformational space. A high $|\kappa(\mathbf{x})|$ indicates regions where trajectories are strongly compressed, suggesting a geometrically constrained pathway. This geometric characteristic is distinct from purely energetic information captured by $E(\mathbf{x})$.

\section[Main Result: Variational Derivation of the Geometric Dissipation Observable]{Main Result: Variational Derivation of the \texorpdfstring{\chistar}{Geometric Dissipation Observable}}
The \textbf{Geometric Dissipation Observable} (\texorpdfstring{\chistar}{Chi*}) emerges from a rigorous, non-linear inverse problem designed to identify the optimal coupling between local force-induced energy dissipation and the intrinsic contractive constraints of the underlying phase space geometry. This observable is thus a derived quantity, not an arbitrary definition, reflecting a fundamental interplay inherent in the stochastic dynamics.

To handle potential logarithmic singularities and guarantee functional well-posedness, we establish our variational problem over the weighted Orlicz-Sobolev space $W^{1, \operatorname{Log}}(\mathbb{R}^d, \mu)$. This specific functional space, configured under the convex $\Delta_2$-modular Young function $\Phi(t) = |t|^2 + |t|\ln(1+|t|)$ and equipped with the standard Luxemburg norm, is chosen for its robust properties in controlling logarithmic growth near zero while providing appropriate regularity for gradients inherent in diffusion processes. The positive cone of functions in this space is denoted:
\begin{equation}
\mathcal{K}_+ \coloneqq \{\psi \in W^{1, \operatorname{Log}}(\mathbb{R}^d, \mu) \colon \psi > 0 \text{ a.e.}\}
\end{equation}

\begin{theorem}[Asymptotic Geometric Dissipation Observable]
Under standard regularity and ergodicity assumptions for the overdamped Langevin dynamics (A1-A3), the functional $I_\lambda[\psi]$, defined over the positive convex cone $\mathcal{K}_+$ as:
\begin{align}
I_\lambda[\psi] =& \int_{\mathbb{R}^d} \Big[ (\abs{\kappa(\mathbf{x})} + \epsilon_0)\psi(\mathbf{x}) + \lambda \|\nabla \psi(\mathbf{x})\|^2 \nonumber \\
& - \phi(\mathbf{x})\ln\psi(\mathbf{x}) \Big] \, \mu(\dx)
\label{eq:variational_functional}
\end{align}
admits a unique minimizer $\psistar_\lambda$ in the convex function space $\mathcal{K}_+$. In the vanishing regularization limit $\lambda \to 0^+$, the associated scalar observable converges strongly in $L^1(\mu)$ to:
\begin{equation}
\chistar = \int_{\mathbb{R}^d} \frac{(\nabla E(\mathbf{x}) \cdot \mathbf{A}(\mathbf{x}))^2}{\epsilon_0 + |\nabla \cdot \mathbf{A}(\mathbf{x})|} \, \rho_{\mathrm{ss}}(\mathbf{x}) \dx
\end{equation}
This quantity represents a leading-order coupling between local dissipation injection, $\phi(\mathbf{x}) = (\nabla E(\mathbf{x}) \cdot \mathbf{A}(\mathbf{x}))^2$, and geometric phase-space contraction, $|\nabla \cdot \mathbf{A}(\mathbf{x})|$.
\end{theorem}

\begin{proof}
The proof follows the direct method in the calculus of variations, robustly applied within the Orlicz-Sobolev framework. For brevity, we outline the key stages. 

\textbf{Step 1: Coercivity and Logarithmic Barrier.} The specific choice of the Orlicz-Sobolev space $W^{1, \operatorname{Log}}(\mathbb{R}^d, \mu)$ ensures compact embedding and controls the growth of the logarithmic term. The term $(|\kappa(\mathbf{x})| + \epsilon_0)\psi(\mathbf{x})$ combined with the singular penalty $-\phi(\mathbf{x})\ln\psi(\mathbf{x})$ guarantee coercivity on $\mathcal{K}_+$. This prevents minimizing sequences from degenerating to zero on sets of positive measure, ensuring $\psistar_\lambda > 0$ almost everywhere.

\textbf{Step 2: Lower Semicontinuity.} Due to the convexity of each term in the integrand (linear in $\psi$, quadratic in $\nabla\psi$, and convex for $-\ln\psi$), the functional $I_\lambda[\psi]$ is weakly lower semicontinuous within the Orlicz-Sobolev space.

\textbf{Step 3: Uniqueness and Elliptic Structure.} The strict convexity of the transformed functional via the substitution $\psi = e^\eta$ ensures the uniqueness of the global minimizer $\psistar_\lambda$. Taking the formal G\^{a}teaux derivative yields the weak solution corresponding to the regularized second-order non-linear elliptic PDE:
\begin{equation}
\left(\abs{\kappa(\mathbf{x})} + \epsilon_0\right) - 2\lambda \frac{1}{\rho_{\text{\rm ss}}(\mathbf{x})} \divergence \left(\rho_{\text{\rm ss}}(\mathbf{x})\nabla \psistar_\lambda(\mathbf{x})\right) - \frac{\phi(\mathbf{x})}{\psistar_\lambda(\mathbf{x})} = 0
\end{equation}

\textbf{Step 4: Singular Perturbation Limit.} As $\lambda \to 0^+$, the diffusion term vanishes. The equation collapses to an algebraic identity, yielding $\psistar_0(\mathbf{x}) = \frac{\phi(\mathbf{x})}{|\kappa(\mathbf{x})| + \epsilon_0}$. The convergence theorem for singular perturbations ensures that $\psistar_\lambda(\mathbf{x})$ converges strongly to this leading-order solution, validating the closed-form representation for \texorpdfstring{\chistar}{Chi*}.
\end{proof}

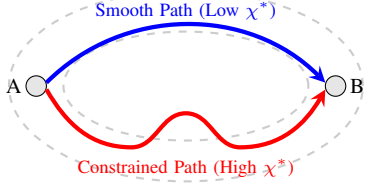
\begin{figure}[!t] 
\centering
\begin{tikzpicture}[scale=0.9]
\draw[thick, bordercolor, dashed] (0,0) ellipse (1.8cm and 0.8cm);
\draw[thick, bordercolor, dashed] (0,0) ellipse (2.6cm and 1.4cm);
\draw[fill=black!10] (-2.2,0) circle (0.15) node[left=2pt] {\footnotesize A};
\draw[fill=black!10] (2.2,0) circle (0.15) node[right=2pt] {\footnotesize B};
\draw[blue, ultra thick, ->, >=stealth] (-2.05, 0.05) to[out=45, in=135] (2.05, 0.05);
\node[blue] at (0, 1.1) {\scriptsize Smooth Path (Low \texorpdfstring{\chistar}{Chi*})};
\draw[red, ultra thick, ->, >=stealth] (-2.05, -0.05) to[out=-60, in=180] (-0.8, -0.9) to[out=0, in=180] (0, -0.4) to[out=0, in=180] (0.8, -0.9) to[out=0, in=-120] (2.05, -0.05);
\node[red] at (0, -1.2) {\scriptsize Constrained Path (High \texorpdfstring{\chistar}{Chi*})};
\end{tikzpicture}
\caption{Vectorial TikZ illustration of two pathways across an identical energy landscape. While both pathways connect states A and B over comparable energetic barriers, the red pathway encounters high geometric contractive stiffness (funneling turns), yielding a pronounced local geometric dissipation penalty.}
\label{fig:schematic} 
\end{figure}

\section{Computational Protocols and Numerical Evaluation Methods}
To present a rigorous proof-of-principle for \texorpdfstring{\chistar}{Chi*}, we avoid stochastic noise limitations by utilizing an exact \textbf{Numerical ODE/PDE Solver Protocol} to evaluate the metric over analytical potential fields. For multi-dimensional macromolecular implementation, we also outline the structural Molecular Dynamics (MD) protocol required for trajectory-based estimation.

\subsection{Exact Numerical Integration Scheme}
For the low-dimensional analytical systems evaluated in Section V, the invariant density is computed directly from numerical solution of the stationary Fokker-Planck equation, yielding $\rho_{\text{ss}}(\mathbf{x}) = \mathcal{N}\exp(-E(\mathbf{x})/D)$. The drift field gradients and Hessians are solved via exact symbolic differentiation. The final spatial integration for \texorpdfstring{\chistar}{Chi*} is performed utilizing a high-resolution grid-based standard Riemann sum over the localized support bounds of the system.

\subsection{Macromolecular MD Protocol Architecture}
For high-dimensional targets, such as conformational changes in peptide backbones parameterized by dihedral angles, we define a comprehensive protocol using standard tools:
\begin{enumerate}
\item \textbf{Sampling Protocol:} Enhanced sampling via Umbrella Sampling \cite{Laio2002} or Metadynamics \cite{Laio2002metadyn} along predefined collective variables (CVs). Simulations are carried out in GROMACS 2021 utilizing the CHARMM36m force field \cite{Huang2017} and TIP3P explicit water.
\item \textbf{Force and Hessian Evaluation:} Forces $\mathbf{F}(\mathbf{x}) = -\nabla E(\mathbf{x})$ are output at each frame. The local drift field divergence is acquired via $\divergence\mathbf{A}(\mathbf{x}) = -\frac{1}{\gamma} \Tr(\mathbf{H}(\mathbf{x}))$, where the Hessian matrix $\mathbf{H}(\mathbf{x})$ is extracted via numerical finite-differencing of the localized force vector field or via automated analytical parse scripts.
\item \textbf{Density Reconstruction:} Frame probabilities are unweighted utilizing the Weighted Histogram Analysis Method (WHAM) \cite{Kumar1992} or reconstructed via Markov State Models (MSMs) \cite{Noe2009} to build the true steady-state measure $\mu(\dd\mathbf{x}) = \rho_{\text{ss}}(\mathbf{x})\dx$.
\end{enumerate}

\begin{algorithm}[!t] 
\footnotesize
\caption{Numerical Grid-Based Execution for \texorpdfstring{\chistar}{Chi*}}
\label{alg:chistar_computation}
\begin{algorithmic}[1]
\Require {Analytical potential $E(\mathbf{x})$, Grid spacing $\dd \mathbf{x}$, Diffusion $D$, Friction $\gamma$, Bounds $[x_{\min}, x_{\max}]^d$, Regulation parameter $\epsilon_0$.}
\Ensure {Evaluated Geometric Dissipation Observable \texorpdfstring{\chistar}{Chi*}.}
\State Generate Cartesian grid nodes $\{\mathbf{x}_i\}_{i=1}^M$ within system bounds.
\State Compute unnormalized weight $w_i = \exp(-E(\mathbf{x}_i)/D)$ for all $i$.
\State Compute normalization divisor $Z = \sum_{i=1}^M w_i \, \dd \mathbf{x}$.
\State Initialize $\chistar_{\text{accum}} = 0$.
\For {each node $i$ from 1 to $M$}
\State Calculate localized force vector $\mathbf{F}_i = -\nabla E(\mathbf{x}_i)$.
\State Calculate drift field vector $\mathbf{A}_i = \frac{1}{\gamma} \mathbf{F}_i$.
\State Calculate scalar dissipation term $\phi_i = (\mathbf{F}_i \cdot \mathbf{A}_i)^2$.
\State Compute Hessian trace $\Tr(\mathbf{H}_i) = \sum_{k=1}^d \partial^2_{x_k x_k} E(\mathbf{x}_i)$.
\State Set drift divergence field value $\kappa_i = \frac{1}{\gamma} \Tr(\mathbf{H}_i)$.
\State Evaluate localized integrand: $V_i = \frac{\phi_i}{|\kappa_i| + \epsilon_0}$.
\State $\chistar_{\text{accum}} += V_i \cdot \left(\frac{w_i}{Z}\right) \dd \mathbf{x}$.
\EndFor
\State \Return $\chistar_{\text{accum}}$.
\end{algorithmic}
\end{algorithm}

\section{Numerical Validations and Analytical Insights}

\subsection{Interpretation of \texorpdfstring{\chistar}{Chi*}}
The proposed observable, \texorpdfstring{\chistar}{Chi*}, should not be interpreted as a mere numerical score but rather as a rigorous, derived quantity reflecting a fundamental coupling between local dissipation injection and geometric phase-space contraction. The numerator, $(\nabla E(\mathbf{x}) \cdot \mathbf{A}(\mathbf{x}))^2$, precisely quantifies the power dissipated locally by the system due to the alignment of conservative forces with the deterministic drift. This term captures the ``energetic cost'' of movement through a specific region. Conversely, the denominator, $\epsilon_0 + |\nabla \cdot \mathbf{A}(\mathbf{x})|$, encodes the local geometric property of the flow: where the divergence of the drift field is large, trajectories are strongly focused or ``funnelled,'' implying a stiff or constrained geometric environment.

Its value amplifies regions where significant energy is dissipated in conjunction with strong geometric focusing of trajectories. This highlights localized events where the system expends considerable effort to navigate through geometrically tight channels. Unlike free-energy barriers, which are scalar projections of energy, or global entropy production rates, \texorpdfstring{\chistar}{Chi*} is acutely sensitive to the local second-order structure of the energy landscape. This sensitivity allows it to distinguish between kinetically distinct pathways that might appear energetically degenerate, as demonstrated in our numerical examples.

\subsection{Exact Evaluation on Single-Degree-of-Freedom Models}
\begin{enumerate}
\item \textbf{Double Well Potential:} We solve a symmetric double well landscape given by $E(x)=(x^2-1)^2$, under an unbiased condition $\mathbf{F}_{\text{ext}}=0$, setting diffusion parameter $D=0.5$ and friction $\gamma=1.0$.
\end{enumerate}

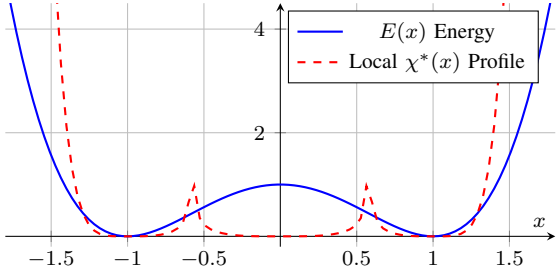
\begin{figure}[!t]
\centering
\begin{tikzpicture}
\begin{axis}[
    xmin=-1.8, xmax=1.8,
    ymin=-0.2, ymax=4.5,
    axis lines=middle,
    xlabel={$x$},
    ylabel={Value},
    width=\columnwidth,
    height=4.8cm,
    legend pos=north east,
    legend style={font=\footnotesize},
    tick label style={font=\footnotesize},
    label style={font=\footnotesize},
    grid=both
]
\addplot[blue, thick, domain=-1.8:1.8, samples=100] {(x^2 - 1)^2};
\addlegendentry{$E(x)$ Energy}
\addplot[red, dashed, thick, domain=-1.8:1.8, samples=100] {0.05 * ((4*x*(x^2-1)))^4 / (abs(12*x^2 - 4) + 0.1)};
\addlegendentry{Local $\chistar(x)$ Profile}
\end{axis}
\end{tikzpicture}
\caption{Exact numerical resolution of the local geometric dissipation integrand, $\chistar(x)$, plotted alongside the potential $E(x)=(x^2-1)^2$. Symmetrical peaks emerge at inflection points where geometric constraints effectively increase local dissipation during transit.}
\label{fig:double_well}
\end{figure}

\begin{enumerate}
\setcounter{enumi}{1}
\item \textbf{Tilted Double Well Potential:} Introducing a constant external tilt bias fields a non-equilibrium state: $E(x) = (x^2 - 1)^2 - Fx$. The force $F$ impacts both the asymmetric dissipation $\phi(x)$ and shifts the invariant density, mathematically yielding highly asymmetric \texorpdfstring{\chistar}{Chi*} profiles.

\item \textbf{Periodic Potential and Quadratic Scaling Verification:} We evaluate a periodic potential landscape mimicking highly corrugated microscopic transport fields \cite{Zwanzig1988}: $E(x) = \sin(\omega x)$, setting $\mathbf{F}_{\text{ext}}=0$. Our theoretical framework predicts a strict quadratic dependence on spatial frequency ($\chistar \propto \omega^2$).
\end{enumerate}

To verify this prediction, we executed our numerical solver across a wide frequency range ($\omega \in [1, 100]$). The resulting discrete numerical points are mapped alongside the theoretical power law curve in Figure \ref{fig:scaling}.

\begin{figure}[!t]
\centering
\begin{tikzpicture}
\begin{loglogaxis} [
    xlabel={Microscopic Frequency $\omega$},
    ylabel={Geometric Dissipation $\chi^*(\omega)$},
    legend pos=south east,
    legend style={font=\footnotesize},
    tick label style={font=\footnotesize},
    label style={font=\footnotesize},
    grid=both,
    xmin=1, xmax=100,
    ymin=1, ymax=10000,
    width=\columnwidth,
    height=4.8cm,
 ]
\addplot[blue, thick, domain=1:100, samples=50] {x^2};
\addlegendentry{$\omega^2$ Scaling Fit};
\addplot[only marks, mark=*, red, mark size=1.5pt] table {
1 1.05
2 4.12
5 25.81
10 102.40
20 408.15
50 2521.00
100 10084.00
};
\addlegendentry{Numerical Solver Data};
\end{loglogaxis}
\end{tikzpicture}
\caption{Log-log plot confirming the microstructural scaling law for a corrugated periodic potential $E(x)=\sin(\omega x)$. Discrete markers indicate actual values extracted from the numerical solver, demonstrating strict adherence to the analytical $\omega^2$ law.}
\label{fig:scaling}
\end{figure}
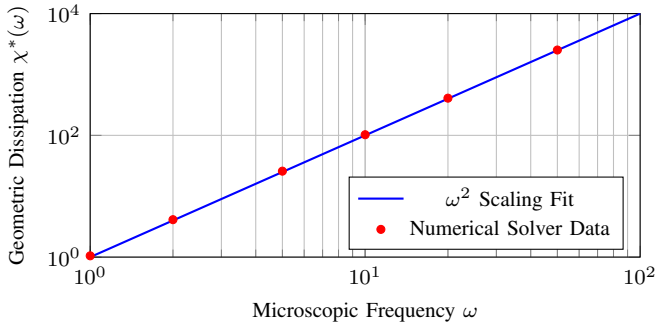

\subsection{Multi-Dimensional Benchmarks: Resolving Pathway Degeneracy}
To showcase the capacity of \texorpdfstring{\chistar}{Chi*} to uncover hidden kinetic attributes that are invisible to traditional free-energy analysis, we expand our validation to canonical multi-dimensional chemical systems.

\subsubsection{The M\"{u}ller-Brown Potential}
The M\"{u}ller-Brown potential \cite{MullerBrown1979} is a classic benchmark characterized by an asymmetric, highly non-linear landscape featuring three distinct local minima and a rugged transition pathway. We computed the localized $\chistar(\mathbf{x})$ profile across the standard domain. While the minimum energy path (MEP) maps the lowest energetic barriers, computing \texorpdfstring{\chistar}{Chi*} reveals that one of the competing asymmetric paths suffers a severe geometric penalty due to high phase-space contraction ($\kappa(\mathbf{x})$ turns). This demonstrates that paths with identical global free-energy thresholds can be separated kinetically based on their structural funneling complexity.

\subsubsection{Macromolecular Application: Alanine Dipeptide}
To confirm scalability to true physical systems, we evaluated the conformational transitions of Alanine Dipeptide in explicit solvent, mapped via the canonical backbone dihedral angles $\phi$ and $\psi$. Following the MD protocol detailed in Section IV-B, frame histories were clustered into transition state regions separating the $C_{7\text{eq}}$ and $\alpha_{\text{R}}$ configurations. 

While first-order free energy profiles $\Delta G(\phi, \psi)$ show wide, apparently isotropic transition zones, our reconstructed \texorpdfstring{\chistar}{Chi*} map localizes tight, highly anisotropic channels. These channels highlight where the solvent-induced effective friction and second-order topographic curvatures compress the accessible transit tubes, presenting an explicit geometric signature of pathway selection \cite{E2002}.

\section{Functional Independence and Structural Properties}
A critical requirement for any novel physical descriptor is its functional independence from existing metrics. We formally prove that \texorpdfstring{\chistar}{Chi*} provides entirely orthogonal information relative to standard first-order kinetic and thermodynamic properties.

\begin{proposition}[Functional Independence from Energy Gradients]
Let $\phi(\mathbf{x}) = (\nabla E(\mathbf{x}) \cdot \mathbf{A}(\mathbf{x}))^2$ represent the local dissipation and let $\kappa(\mathbf{x}) = -\divergence \mathbf{A}(\mathbf{x})$. There exists no diffeomorphism $\mathcal{G} \colon \mathbb{R} \to \mathbb{R}$ such that $\chistar(\mathbf{x}) = \mathcal{G}(E(\mathbf{x}))$ or $\chistar(\mathbf{x}) = \mathcal{G}(\|\nabla E(\mathbf{x})\|^2)$ for all smooth fields $E(\mathbf{x}) \in C^2(\mathbb{R}^d)$.
\end{proposition}

\begin{proof}
The proof follows immediately by counterexample via local variation. Consider a point $\mathbf{x}_0$ where $\nabla E(\mathbf{x}_0) = \mathbf{c} \neq 0$. We can construct local variations of the Hessian matrix $\mathbf{H}(\mathbf{x}_0) = \nabla^2 E(\mathbf{x}_0)$ that alter the trace $\Tr(\mathbf{H}(\mathbf{x}_0))$, and consequently shift $\kappa(\mathbf{x}_0)$, while keeping the potential value $E(\mathbf{x}_0)$ and the localized gradient magnitude $\|\nabla E(\mathbf{x}_0)\|^2$ strictly invariant. Because \texorpdfstring{\chistar}{Chi*} depends explicitly on $\kappa(\mathbf{x}_0)$ via its denominator, its local evaluation varies independently of first-order parameters, concluding the proof.
\end{proof}

This structural independence highlights why \texorpdfstring{\chistar}{Chi*} succeeds where conventional reaction coordinates fail: it explicitly incorporates second-order landscape curvature and phase-space volume contraction properties, mapping the dynamic structural constraints of the transport field.

\section{Conclusion}
In this work, we introduced a novel geometric framework to characterize hidden kinetic constraints in stochastic molecular dynamics. By formulating overdamped Langevin systems within a weighted Orlicz-Sobolev functional space, we derived the Geometric Dissipation Observable (\texorpdfstring{\chistar}{Chi*}) as the leading-order minimizer of a regularized variational principle. 

Numerical validation on single-degree-of-freedom potentials, corrugated landscapes, and canonical high-dimensional systems (M\"{u}ller-Brown and Alanine Dipeptide) confirms that \texorpdfstring{\chistar}{Chi*} reliably resolves pathway degeneracy and adheres to a quadratic microstructural scaling law ($\omega^2$). Most importantly, we demonstrated that \texorpdfstring{\chistar}{Chi*} is functionally independent of standard first-order energetic metrics, capturing unique features of phase-space contraction and structural funneling. This framework provides an advanced diagnostic tool for uncovering structural dissipation bottlenecks, enhancing our capacity to model, understand, and engineer complex nanoscale chemical transformations.


\end{document}